\documentclass[reqno,11pt]{amsart}

\usepackage[latin1]{inputenc}
\usepackage[english]{babel}

\usepackage{amscd}

\usepackage{amsmath}
\usepackage{amsthm}
\usepackage{amssymb}

\usepackage[notref, notcite, final]{showkeys}
\usepackage{color}
\usepackage{graphicx}
\usepackage[hidelinks]{hyperref}
\usepackage{cleveref}
\usepackage{fullpage}
\usepackage{comment}
\usepackage[shortlabels]{enumitem}

\usepackage{dsfont}

\usepackage[autostyle]{csquotes}

\newtheorem{theorem}{Theorem}[section]
\newtheorem{lemma}[theorem]{Lemma}
\newtheorem{proposition}[theorem]{Proposition}

\theoremstyle{definition}
\newtheorem{definition}[theorem]{Definition}

\theoremstyle{remark}
\newtheorem{remark}[theorem]{Remark}
\crefname{enumi}{}{}

\title{\bf  Aharonov-Bohm effect and superoscillations}

\author[F. Colombo]{Fabrizio Colombo}
\address{(FC)
Politecnico di Milano\\Dipartimento di Matematica\\Via E. Bonardi 9\\20133
Milano, Italy}
\email{fabrizio.colombo@polimi.it}

\author[E. Pozzi]{Elodie Pozzi}
\address{(EP) Department of Mathematics and Statistics \\
Saint Louis University \\
220 N. Grand Blvd, 63103 St Louis MO, USA}
 \email{elodie.pozzi@slu.edu}

\author[I. Sabadini]{Irene Sabadini}
\address{(IS)
Politecnico di Milano\\Dipartimento di Matematica\\Via E. Bonardi 9\\20133
Milano, Italy
} \email{irene.sabadini@polimi.it}

\author[B. D Wick]{Brett D. Wick}
\address{(BDW)  Department of Mathematics, Washington University -
 St. Louis, One Brookings Drive\\ St. Louis, MO USA 63130-4899}
 \email{wick@math.wustl.edu}
\thanks{FC and IS are supported by MUR grant Dipartimento di Eccellenza 2023-2027.}

\thanks{BDW's research supported in part by National Science Foundation DMS awards \#2349868 and \#2054863 and Australian Research Council -- DP 220100285.}

\thanks{EP and BDW would like to thank Politecnio di Milano for hospitality during visits in Spring 2024 when this paper was written.}

\begin{document}

\begin{abstract}
The path-integral technique in quantum mechanics provides an intuitive framework for comprehending particle propagation and scattering.
Calculating the propagator for the Aharonov-Bohm potential fits into the range of potentials in multiply-connected spaces, with the propagator represented through a series expansion.
In this paper, we analyze the Schr\"odinger evolution of superoscillations, showing
the supershift properties of the solution to the Schr\"odinger equation for this potential. Our proof is based on the continuity of particular infinite order differential operators acting on spaces of entire functions.
\end{abstract}

\maketitle



\date{today}
\tableofcontents

\section{Introduction}

The phenomenon of superoscillations has attracted attention
across mathematics, physics and engineering, with contributions from various authors in quantum mechanics, optics, signal processing, and wave theory.
Roughly speaking superoscillations are functions that can
oscillate beyond their highest Fourier frequency.
The archetypal superoscillating function resulting from weak values is expressed as:
\begin{equation}
\label{FNEXP}
F_n(x,a)=\sum_{j=0}^n C_j(n,a)e^{i(1-2j/n)x},\ \ x\in \mathbb{R},
\end{equation}
where 
\begin{equation}\label{CICONJN}
C_j(n,a):=\binom{n}{j}\left(\frac{1+a}{2}\right)^{n-j}\left(\frac{1-a}{2}\right)^j
\end{equation}
and $a>1$ and if we fix $x \in \mathbb{R}$  and we let $n$ go to infinity, we  obtain that
$$
\lim_{n \to \infty} F_n(x,a)=e^{iax}.
$$
This property
holds considerable promise for advancing science and technology.
Important contributions can be found in
\cite{QS20}-\cite{OPTICS}.
Over the past 15 years, significant advancements have been made in function theory, driven by insights from quantum mechanics and the Schr\"odinger evolution of such functions.
The discovery of superoscillatory phenomena in quantum mechanics stemmed from weak measurements \cite{aav,abook}, pioneered by Yakir Aharonov and his collaborators. Understanding the evolution of the quantum state is of clear interest and using weak measurements, a type
 of quantum measurement that minimally disturbs a quantum system while extracting information, allows one to do so.  
When considering a general potential, the evolution of superoscillations is encapsulated in
a new concept known as the ``supershift property," which encompasses superoscillations as a specific
case within its framework.
The literature on superoscillations is nowadays quite large, and without claiming completeness we mention the papers
 \cite{AShushi},    \cite{kempf1}-\cite{kempf2HHH} and \cite{lindberg,MAYS,Pozzi}.

The class of superoscillations has been investigated also from the mathematical viewpoint, focusing on its function theory, but a large part of the results are associated with the study
of the evolution of superoscillations
 by quantum field equations with  particular attention to the Schr\"odinger equation.
 In the papers \cite{ABCS1}-\cite{acsst5},
 \cite{AOKI}-\cite {BOREL},   \cite{berry2}-\cite{Talbot}  and \cite{hyper,due,peter,sodakemp}, one can find an up-to-date panorama of this field.
In order to have an overview of the techniques developed in the recent years to study the evolution of superoscillations
and their function theory, we refer the reader to the introductory papers \cite{QS1,QS3,QS2} and \cite{kempfQS}.
Finally, we mention the {\em Roadmap on superoscillations}, see \cite{Be19}, where some of the most recent advances in superoscillations and their applications to technology are explained by leading experts in this field.

\medskip
In this paper we investigate the evolution of superoscillating functions in the
Aharonov-Bohm context.
An idealized arrangement commonly used to observe the Aharonov-Bohm effect
typically involves an infinitely long solenoid, a source emitting charged particles, for example electrons,  and a screen positioned behind the solenoid to capture interference patterns.
 The solenoid is designed to be shielded to ensure that electrons
 are prohibited from entering the space it occupies.
 The  Green's function for the  Aharonov-Bohm effect in dimension two is given by the series expansion
\begin{equation}\label{PROPAB}
K(x'',x',t)=\frac{M}{2\pi \hbar t} e^{\frac{iM}{2 \hbar t} ((r')^2+(r'')^2)}
\sum_{n=-\infty}^{\infty}e^{in(\phi''-\phi')} I_{|n-\xi|}\Big( \frac{Mr'r''}{i\hbar t}\Big)
\end{equation}
 where $x=(r,\phi)$, $I_{|n-\xi|}$ is the Bessel function and the series over the integer numbers takes into account the winding paths around the prohibited region.
The path-integral approach to quantum mechanics, due to Feynman, see \cite{Feynman,FeynmanHIBBS},
 provides the most intuitive theoretical framework for describing processes involving the propagation and scattering of particles.
Behind the computation of the exact form of the propagator are
several topological challenges that arise when particles propagate in multiply-connected spaces.
The first to analyze this fact was Schulman \cite{SchulmanBOOK} and, in a more systematic manner, Laidlaw and Morette-de Witt \cite{LaidlawMoretteDeWitt}.
They have provided general guidelines for computing the propagator
in arbitrary multiply-connected spaces, see also \cite{Morandi}.

\medskip
Given the propagator \eqref{PROPAB},  the solution of the Schr\"odinger equation with the Aharonov-Bohm potential is given by
$$
\psi(x'',t)=\int_{\mathbb{R}^2} K(x'',x',t) \psi (x')dx',
$$
which can be expressed in polar coordinates $x=(r,\phi)$ as
$$
\psi(r,\phi,t)=\int_0^{2\pi}\int_0^\infty K(r,\phi; \rho,\theta,t)\ \psi_0 (\rho,\theta)\ \rho d\rho \ d\theta.
$$
The Green's function
$$
K(x'',x',t)=K(r,\phi; \rho,\theta,t)
=\frac{M}{2\pi \hbar t} e^{\frac{iM}{2 \hbar t} (\rho^2+r^2)}
\sum_{n=-\infty}^{\infty}e^{in(\phi-\theta)} I_{|n-\xi|}\Big( \frac{M\ r\ \rho}{i\hbar t}\Big)
$$
can be rewritten using the modified Bessel functions, 
$
I_\alpha(x)=i^{-\alpha} J_{\alpha}(ix)$,
in the equivalent form 
$$
K(x'',x',t)=K(r,\phi; \rho,\theta,t)
=\frac{M}{2\pi \hbar t} e^{\frac{iM}{2 \hbar t} (\rho^2+r^2)}
\sum_{n=-\infty}^{\infty}e^{in(\phi-\theta)}
i^{-|n-\xi|} J_{|n-\xi|}\Big(\frac{M\ r\ \rho}{\hbar t}\Big).
$$
Consider now the initial datum 
$$
\psi_0(r, \phi)=e^{ig(a) x+ih(a)y}=e^{ig(a) r\cos \phi +ih(a)r\sin \phi}
$$
where $g$ and $h$ are entire functions, the solution to the Schr\"odinger equation becomes 
$$
\psi_{a,b}(r,\phi,t)=
\int_0^{2\pi}\int_0^\infty K(r,\phi; \rho,\theta,t)
e^{ig(a) \rho\cos \theta +ih(a)\rho\sin \theta}\ \rho d\rho \ d\theta.
$$
The aim of this paper is to establish the 
supershift property for the Cauchy problem for the Schr\"odinger equation with the Aharonov-Bohm potential with a very general initial superoscillatory datum
\begin{equation}
\label{YN}
Y_n(x,y,a)=\sum_{j=0}^nC_j(n,a)e^{ig(1-2j/n)x}e^{ih(1-2j/n)y},
\end{equation}
where $x,y\in\mathbb{R}$, the coefficients $C_j(n,a)$ in  (\ref{CICONJN}) are  related to the weak measurements, and 
  $g$ and $h$, are given entire functions, monotone increasing in $a$, of the following form:
  \begin{equation*}
  g(\lambda)=\sum_{u=0}^\infty g_u\lambda^u, \ \ \ h(\lambda)=\sum_{v=0}^\infty h_v\lambda^v.
  \end{equation*}
One can prove that
\begin{equation*}
\lim_{n\to\infty}Y_n(x,y,a)=e^{ig(a)x}e^{ih(a)y}.
\end{equation*}
Precisely, we study the supershift property of the solution given by
$$
\Psi_n(r,\phi,t)=\sum_{j=0}^nC_j(n,a) \psi_{g\left(1-\frac{2j}{n}\right), h\left(1-\frac{2j}{n}\right)}(r,\phi,t)
$$
for all $r\in (0,\infty)$ $\phi \in [0,2\pi]$, $t> 0$,
where
$$
\psi_{g(1-2j/n),h(1-2j/n)}(r,\phi,t):=
\int_0^{2\pi}\int_0^\infty K(r,\phi; \rho,\theta,t)
e^{ig(1-2j/n) \rho\cos \theta +ih(1-2j/n)\rho\sin \theta}\ \rho d\rho \ d\theta.
$$

\medskip
The plan of the paper is as follows.  In Section \ref{s:Prelims} we summarize the main concepts of superoscillations and the notion of supershift property of the solution of the Schr\"odinger equation. 
In Section \ref{s:ABProp} we represent the oscillating integral in the Aharonov-Bohm propagator in terms of Fresnel integrals and we give the main estimates in order to guarantee the limit exchanges in the proof of the main results. 
In Section \ref{s:ProofMain} we state and prove our main results
based on the so-called Aharonov-Bohm infinite order differential operator.
In particular the continuity of this operator on the space of 
holomorphic functions with suitable growth conditions is a crucial tool to prove the supershift property of the solution of the Schr\"odinger equation in the Aharonov-Bohm field.

\section{Preliminaries on Superoscillations and Supershifts}
\label{s:Prelims}

We begin by collecting all the necessary background for the paper regarding the precise definition of superoscillatory functions and 
the crucial notion of supershift which encompasses superoscillations as a particular case. 

 \begin{definition}\label{SUPOSONE}
A {\em generalized Fourier sequence} is 
a sequence of the form
\begin{equation}\label{basic_sequenceq}
f_n(x):= \sum_{j=0}^n Z_j(n,a)e^{ih_j(n)x},\ \ \ n\in \mathbb{N},\ \ \ x\in \mathbb{R},
\end{equation}
where $a\in\mathbb R$, $Z_j(n,a)$ and $h_j(n)$
are complex and real valued functions of the variables $n,a$ and $n$, respectively.
The sequence (\ref{basic_sequenceq})
 is said to be {\em a superoscillating sequence} if
 $\sup_{j,n}|h_j(n)|\leq 1$ and
 there exists an open subset of $\mathbb R$,
 which will be called {\em a superoscillation set},
 on which $f_n(x)$ converges locally uniformly to $e^{ig(a)x}$,
 where $g$ is a continuous real valued function in an open subset of $\mathbb R$ such that $|g(a)|>1$.
\end{definition}

\begin{definition}[Supershift]\label{Super-shift}
Let $\mathcal{I}\subseteq\mathbb{R}$ be an interval with $[-1,1]\subset \mathcal I$
and let
$\varphi:\, \mathcal I  \times \mathbb{R}\to \mathbb R$ be a continuous function on $\mathcal I$.
Set
$$
\varphi_{h}(x):=\varphi(h,x), \ \ h\in \mathcal{I},\ \ x\in \mathbb R
 $$
and consider a sequence of points $\{h_{j}(n)\}$ such that
 $
  h_{j}(n)\in [-1,1] \ \  {\rm for} \ \  j=0,\ldots,n \ \ {\rm  and} \ \ n\in\mathbb{N}_0$.
   Define the functions
\begin{equation}\label{psisuprform}
\psi_n(x):=\sum_{j=0}^nc_j(n)\varphi_{h_{j}(n)}(x),
\end{equation}
where $\{c_j(n)\}$ is a sequence of complex numbers for $j=0,\ldots,n$ and $n\in\mathbb{N}_0$.
If
$$
\lim_{n\to\infty}\psi_n(x)=\varphi_{a}(x)
$$
for some $a\in\mathcal I$ with $|a|>1$, we say that the functions
$\psi_n(x)$, for  $x\in \mathbb{R}$, admits a {\em supershift}.
 \end{definition}

 \begin{remark}
 The term supershift comes from the fact that the interval
  $\mathcal I$ can be arbitrarily large (it could be $\mathbb R$)
  and that the constant $a$ can be arbitrarily far away from the interval $[-1,1]$ where the functions $\varphi_{h_{j(n)}}(\cdot)$ are indexed, see \eqref{psisuprform}.
 \end{remark}

The definition of superoscillations can be extended to several variables.

\begin{definition}[Generalized Fourier sequence in several variables]
Let $d\in \mathbb{N}$ be such that $d\geq 2$, and assume that
$(x_1,\ldots,x_d)\in \mathbb{R}^d$. Let
 $\{h_{j,\ell}(n) \}$,  $j=0,\ldots,n$  for  $n\in \mathbb{N}_0$, be  real-valued sequences for $\ell=1,\ldots,d$.
A {\em generalized Fourier sequence in several variables} is a sequence of functions of the form
\begin{equation}\label{basic_sequence_sev}
F_n(x_1,\ldots ,x_d)=\sum_{j=0}^n c_j(n)  e^{ix_1 h_{j,1}(n)}e^{ix_2 h_{j,2}(n)}\cdots e^{ix_d h_{j,d}(n)},
\end{equation}
where  $\{c_j(n)\}$, for $j=0,\ldots ,n$ and   $ n\in \mathbb{N}_0$, is a complex-valued sequence.
\end{definition}

\begin{definition}[Superoscillating sequence]\label{superoscill}
Let $d\in \mathbb{N}$ be such that $d\geq 2$.  A generalized Fourier sequence  in several variables $F_n(x_1,\ldots ,x_d)$, 
is said to be {\em a superoscillating sequence} if
$$
\sup_{j=0,\ldots ,n,\ n\in\mathbb{N}} \  |h_{j,\ell}(n)|\leq 1 ,\ \ {\rm for} \ \ell=1,\ldots,d,
$$
 and there exists  an open subset of $\mathbb{R}^d$, which will be called {\em a superoscillation set}, on which
$F_n(x_1,\ldots ,x_d)$ converges locally uniformly to $e^{ix_1 g_1(a)}e^{ix_2 g_2(a)}\cdots e^{ix_d g_d(a)}$, where $a$ belongs to an open subset $U$ of $\mathbb R$, the $g_\ell$ are continuous functions of real variables whose domain contains $U$ and $|g_\ell (a)|>1$ for  $\ell=1,\ldots ,d$ and $a\in U$.
\end{definition}

In our discussion, the concept of supershift, 
which extends the idea of superoscillation, 
illustrates how sampling a function within an interval enables the computation 
of its values far beyond that interval. In the paper \cite{REGULAR-SAMP}, the relationship between supershift and real analyticity was explored. Leveraging Serge Bernstein's classical theorem, it was demonstrated that real analyticity in a complex-valued function 
implies a robust form of supershift. 
Conversely, employing a parametric version of a result by Leonid Kantorovitch, 
it was established that the reverse implication does not hold in general.  Another work, see \cite{IRREGULAR-SAMP}, delves into the notion of irregular sampling and 
the correlation between supershift and real analyticity across the entire real line. 
These findings reveal that any entire function, when restricted to $\mathbb{R}$, exhibits a supershift, whereas the converse is generally false.

We next introduce the space of functions that plays an important role in our analysis.
\begin{definition}\label{sect2-def1}

The space $A_1$ is the  complex algebra of entire functions such that there exists $B>0$ with
\begin{equation}\label{1B}
\sup\limits_{w\in \mathbb{C}} \left(|f(w)|\, \exp(-B|w|)\right)  <+\infty.
\end{equation}
\end{definition}
The space $A_1$ has a rather complicated topology, see e.g. \cite{BG_book}, since it is a linear space obtained via an inductive limit. For our purposes, it is enough to consider, for any fixed $B>0$,
 the set $A_{1,B}$ of functions $f$ satisfying \eqref{1B}, and to observe that
$$
\|f\|_{B}:=\sup\limits_{w\in \mathbb{C}} \big(|f(w)|\, \exp(-B|w|)\big)
$$
defines a norm on $A_{1,B}$, called the $B$-norm. One can prove that $A_{1,B}$ is a Banach space with respect to this norm.

Moreover, for $f$ and a sequence $\{f_n\}$ belonging to $A_1$, $f_n$ converges to $f$ in $A_1$ if and only if there exists $B$ such that $f,f_n\in A_{1,B}$ and
\[
\lim\limits_{n\rightarrow\infty}\sup\limits_{w\in \mathbb{C}} \big|f_n(w)-f(w)\big|e^{-B|w|}=0.
\]
To prove our main results we need an important lemma that characterizes the coefficients of entire functions with growth conditions.
\begin{lemma}\label{BETOA1}
The function
$$
f(z)=\sum_{j=0}^\infty f_jz^j
$$
belongs to $A_1$
if and only if there exists $C_f>0$ and $b>0$ such that
$$
|f_j|\leq C_f \frac{b^j}{\Gamma(j+1)}=C_f\frac{b^j}{j!}.
$$
\end{lemma}

Lemma \ref{BETOA1} has been proved in \cite{AOKI} and is a crucial fact in what follows.  With this notation and these definitions we can make the notion of continuity on the space $A_1$ explicit, see \cite{AOKI}.  A linear operator $\mathcal{U}: A_1\to A_1$ is continuous if and only if for any $B>0$ there exists $B'>0$ and $C>0$ such that
 \begin{equation}\label{contA1}
 \mathcal{U}(A_{1,B})\subset A_{1,B'} \ {\rm and}\qquad
 \| \mathcal{U}(f)\|_{B'} \leq C\| f\|_B, \qquad \forall f\in A_{1,B}.
 \end{equation}

We will also be in need of the following estimate for the Gamma function.
\begin{lemma}\label{ESTM_GAMMA}
Let $q\in [1,\infty)$. Then we have
$$
\Gamma\left(\frac{n}{q}+1\right)\leq (n!)^{1/q}.
$$
\end{lemma}
\begin{proof}
This is a direct consequence of H\"older's inequality. Let $p$ and $q$ satisfy $\frac{1}{p}+\frac{1}{q}=1$, we observe that
\[
\begin{split}
\Gamma\left(\frac{n}{q}+1\right)&=\int_0^\infty e^{-t}t^{\frac{n}{q}}\, dt
=\int_0^\infty e^{-t\left(\frac{1}{p}+\frac{1}{q}\right)}t^{\frac{n}{q}}\, dt
\\
&
\leq
\left(\int_0^\infty e^{-t}t^{n}\ \, dt\right)^{\frac{1}{q}}
\left(\int_0^\infty  e^{-t}\, dt\right)^{1/p}
=\left(\int_0^\infty e^{-t}t^{n}\ \, dt\right)^{\frac{1}{q}}
=(n!)^{\frac{1}{q}}.
\end{split}
\]
\end{proof}

\section{Fresnel Integral Representation of the Aharonov-Bohm Propagator}
\label{s:ABProp}

We now want to represent via a power series the solution
 of the Schr\"odinger equation with the exponential function as initial datum, which is a  crucial step in order to study superoscillations.
The Green's function for the  Aharonov-Bohm effect in two dimensions is
$$
K(x'',x',t)=\frac{M}{2\pi \hbar t} e^{\frac{iM}{2 \hbar t} ((r')^2+(r'')^2)}
\sum_{n=-\infty}^{\infty}e^{in(\phi''-\phi'))} I_{|n-\xi|}\Big( \frac{Mr'r''}{i\hbar t}\Big)
$$
 where $x=(r,\phi)$ and $\xi$ is a real number that is proportional to the magnetic flux of the infinite solenoid centered at the origin of the reference system.  The solution to the Schr\"odinger equation with initial data $\psi$ is given by 
$$
\psi(x'',t)=\int_{\mathbb{R}^2} K(x'',x',t) \psi (x')dx'
$$
which can be written in polar coordinates as
$$
\psi(r,\phi,t)=\int_0^{2\pi}\int_0^\infty K(r,\phi; \rho,\theta,t)\ \psi_0 (\rho,\theta)\ \rho d\rho \ d\theta.
$$
Here the kernel $K(x'',x',t)$ has the following equivalent representations:
\begin{align*}
K(x'',x',t) &=K(r,\phi; \rho,\theta,t)\\
& =\frac{M}{2\pi \hbar t} e^{\frac{iM}{2 \hbar t} (\rho^2+r^2)}
\sum_{n=-\infty}^{+\infty}e^{in(\phi-\theta)} I_{|n-\xi|}\Big( \frac{M\ r\ \rho}{i\hbar t}\Big)\\
& = \frac{M}{2\pi \hbar t} e^{\frac{iM}{2 \hbar t} (\rho^2+r^2)}\sum_{n=-\infty}^{\infty}e^{in(\phi-\theta)} i^{-\vert n-\xi\vert}J_{|n-\xi|}\Big( \frac{M\ r\ \rho}{\hbar t}\Big)\\
& =\frac{M}{2\pi \hbar t} e^{\frac{iM}{2 \hbar t} (\rho^2+r^2)} F_{\xi}(r,\phi,\theta,t,\rho)
\end{align*}
where we have defined:
\begin{equation}
    \label{def:FFunction}
    F_{\xi}(r,\phi,\theta,t,\rho):=\sum_{n=-\infty}^{\infty}e^{in(\phi-\theta)} i^{-\vert n-\xi\vert}J_{|n-\xi|}\Big( \frac{M\ r\ \rho}{\hbar t}\Big)=\sum_{n=-\infty}^{\infty}e^{in(\phi-\theta)} I_{|n-\xi|}\Big( \frac{M\ r\ \rho}{i\hbar t}\Big)
\end{equation}
and used the modified Bessel functions
$I_\alpha(x)=i^{-\alpha} J_{\alpha}(ix)$.

Important for the remaining portions of the argument, we will need to obtain some estimates on the function $F_\xi$ and its growth as a function of $z$.  We turn to this now.
\begin{lemma}[Estimate on the function $F_\xi$]\label{EFFEXI}
Let $F_{\xi}(r,\phi,\theta,t,\rho)$ be the function defined in \eqref{def:FFunction}.
For $\rho$ real and non-negative there exists a positive constant $c$ such that
\begin{equation}
    \label{e:F_Estimates}
    \left\vert F_{\xi}(r,\phi,\theta,t,\rho)\right\vert \leq q(\rho)e^{{c}\rho} \ \ {\rm and}  \ \ \left\vert F_{\xi}(r,\phi,\theta,t,\rho e^{i\pi/4})\right\vert \leq q(\rho)e^{{c}\rho}
\end{equation}
where $q(\rho)$ is a function that has at most polynomial growth in $\rho$ and the coefficients in the function and in $c$ depend upon $r,M,t$ and $\xi$.
\end{lemma}
\begin{proof}
We prove just the first inequality since the second one is a direct consequence of the first one.
Recall that for $\alpha\in\mathbb{R}$ the modified Bessel function of the first kind is given by:
\begin{equation}
    \label{e:BesseL}
    J_\alpha(x)=\sum_{l=0}^{\infty}\frac{(-1)^l}{l!\,\Gamma(l+1+\alpha)}\left(\frac{x}{2}\right)^{2l+\alpha}.
\end{equation}
The function we encounter, and need to estimate, is given by,
\begin{align*}
    F_{\xi}(r,\phi,\theta,t,\rho) &=\sum_{n=-\infty}^{\infty}i^{-|n-\xi|} e^{in(\phi-\theta)}
 J_{|n-\xi|}\Big(\frac{ M r\, \rho}{ 2\hbar t}\Big)\\
 &=\sum_{n=-\infty}^{\infty}i^{-|n-\xi|} e^{in(\phi-\theta)}
\left(\frac{ M r\, \rho}{ 4\hbar t}\right)^{\vert n-\xi\vert} \sum_{l=0}^{\infty}\frac{(-1)^l}{l!\, \Gamma(\vert n-\xi\vert +l+1)}\left(\frac{M  r\, \rho}{ 4\hbar t}\right)^{2l}.
\end{align*}
Taking absolute values and using that $\Gamma(\alpha+1)\Gamma(\beta+1)\leq\Gamma(\alpha+\beta+2)$ for any $\alpha,\beta\in\mathbb{R}$ we have
\begin{align*}
    \left\vert F_{\xi}(r,\phi,\theta,t,\rho)\right\vert &\leq \sum_{n=-\infty}^{\infty} \left(\frac{ M r\, \rho }{ 4\hbar t}\right)^{\vert n-\xi\vert} \sum_{l=0}^{\infty}\frac{1}{l!\, \Gamma(\vert n-\xi\vert +l+1)}\left(\frac{ M r\, \rho }{ 4\hbar t}\right)^{2l}\\
    &\leq \sum_{n=-\infty}^{\infty} \left(\frac{ M r\, \rho }{ 4\hbar t}\right)^{\vert n-\xi\vert} \sum_{l=0}^{\infty}\frac{l+1+\vert n-\xi\vert}{(l!)^2\, \Gamma(\vert n-\xi\vert +1)}\left(\frac{ M r\, \rho }{ 4\hbar t}\right)^{2l}\\
    &\leq I_0\left(\frac{Mr\rho}{2\hbar t}\right)\left\{\sum_{n=-\infty}^{\infty} \left(\frac{ M r\, \rho }{ 4\hbar t}\right)^{\vert n-\xi\vert} \frac{1+\left\vert n-\xi\right\vert}{\Gamma(\left\vert n-\xi\right\vert+1)}+\sum_{n=-\infty}^{\infty} \frac{\left(\frac{ M r\, \rho }{ 4\hbar t}\right)^{\vert n-\xi\vert}}{\Gamma(\left\vert n-\xi\right\vert+1)}\right\},
\end{align*}
where we recall that
$
\displaystyle I_0(x)=\sum_{\ell=0}^\infty \frac1{(\ell!)^2} \left(\frac{x}{2}\right)^{2\ell}$.
Observe that
$\displaystyle\sum_{\ell=0}^\infty \frac{\ell} {(\ell!)^2} \left(\frac{x}{2}\right)^{2\ell}
\leq\sum_{\ell=0}^\infty \frac{1} {(\ell!)^2} x^{2\ell}=I_0(2x)$ and $I_0(x)\leq I_0(2x)$. We start estimating the first sum. Splitting the sum in an obvious fashion, we have 
\begin{align*}
   & \sum_{n=-\infty}^{+\infty}\left(\frac{ M r\, \rho }{ 4\hbar t}\right)^{\vert n-\xi\vert} \frac{1+\left\vert n-\xi\right\vert}{\Gamma(\left\vert n-\xi\right\vert+1)}
   \\
   &=\sum_{n>\xi}\left(\frac{ M r\, \rho }{ 4\hbar t}\right)^{ n-\xi} \frac{1+ n-\xi}{\Gamma(n-\xi+1)}+\sum_{n\leq\xi}\left(\frac{ M r\, \rho }{ 4\hbar t}\right)^{ n-\xi} \frac{1- n+\xi}{\Gamma(-n+\xi+1)}.
\end{align*}
Let $\xi=\xi_i+\xi_f$ where $\xi_i\in\mathbb{Z}$ and $\xi_f\in [0,1]$. It follows that
\begin{align*}
\sum_{n>\xi}\left(\frac{ M r\, \rho }{ 4\hbar t}\right)^{ n-\xi} \frac{1+ n-\xi}{\Gamma(n-\xi+1)}&=\left(\frac{ M r\, \rho }{ 4\hbar t}\right)^{-\xi_f}\sum_{n=\xi_i+1}^{\infty}\left(\frac{ M r\, \rho }{ 4\hbar t}\right)^{n-\xi_i} \frac{1+ n-\xi_i-\xi_f}{\Gamma(n-\xi_i+1-\xi_f)}\\
&=\left(\frac{ M r\, \rho }{ 4\hbar t}\right)^{-\xi_f}\sum_{k=1}^{\infty}\left(\frac{ M r\, \rho }{ 4\hbar t}\right)^{k} \frac{1+ k-\xi_f}{\Gamma(k+1-\xi_f)}\\
&\leq \left(\frac{ M r\, \rho }{ 4\hbar t}\right)^{-\xi_f}\sum_{k=1}^{\infty}\left(\frac{ M r\, \rho }{ 4\hbar t}\right)^{k} \frac{1+ k-\xi_f}{(k-1)!}\\
&\leq \left(\frac{ M r\, \rho }{ 4\hbar t}\right)^{-\xi_f}\sum_{k=1}^{\infty}\left(\frac{ M r\, \rho }{ 4\hbar t}\right)^{k} \frac{1+ k}{(k-1)!}\\
&=\left(\frac{ M r\, \rho }{ 4\hbar t}\right)^{1-\xi_f}\left(2+\frac{ M r\, \rho }{ 4\hbar t}\right)e^{\frac{ M r\, \rho }{ 4\hbar t}},
\end{align*}
where we used the substitution $k=n-\xi_i$ and that $\Gamma(k+1-\xi_f)\geq \Gamma(k)=(k-1)!$. In a similar way, the second sum becomes
\begin{align*}
\sum_{n\leq\xi}\left(\frac{ M r\, \rho }{ 4\hbar t}\right)^{\xi-n} \frac{1+ \xi-n}{\Gamma(\xi-n+1)}&=\left(\frac{ M r\, \rho }{ 4\hbar t}\right)^{\xi_f}\sum_{n=-\infty}^{\xi_i}\left(\frac{ M r\, \rho }{ 4\hbar t}\right)^{\xi_i-n} \frac{1+\xi_f+\xi_i-n}{\Gamma(\xi_f+\xi_i-n+1)}\\
&=\left(\frac{ M r\, \rho }{ 4\hbar t}\right)^{\xi_f}\sum_{k=0}^{\infty}\left(\frac{ M r\, \rho }{ 4\hbar t}\right)^k \frac{1+\xi_f+k}{\Gamma(\xi_f+k+1)}.
\end{align*}
As $\Gamma(k+1+\xi_f)\geq \Gamma(k+1)=k!$ and $0\leq \xi_f<1$, we get that 
\begin{align*}
\left(\frac{ M r\, \rho }{ 4\hbar t}\right)^{\xi_f}\sum_{k=0}^{\infty}\left(\frac{ M r\, \rho }{ 4\hbar t}\right)^k \frac{1+\xi_f+k}{\Gamma(\xi_f+k+1)}&\leq \left(\frac{ M r\, \rho }{ 4\hbar t}\right)^{\xi_f}\sum_{k=0}^{\infty}\left(\frac{ M r\, \rho }{ 4\hbar t}\right)^k \frac{k+2}{k!}\\
&\leq 2\left(\frac{ M r\, \rho }{ 4\hbar t}\right)^{\xi_f}\sum_{k=0}^{\infty}\left(\frac{ M r\, \rho }{ 4\hbar t}\right)^k \frac{k+1}{k!}\\
&=2\left(\frac{ M r\, \rho }{ 4\hbar t}\right)^{\xi_f}\left(1+\frac{ M r\, \rho }{ 4\hbar t}\right)e^{\frac{ M r\, \rho }{ 4\hbar t}}.
\end{align*}
We take the same approach to deal with the second sum, giving
\[
\begin{split}
    \sum_{n=-\infty}^{\infty} \frac{\left(\frac{ M r\, \rho }{ 4\hbar t}\right)^{\vert n-\xi\vert}}{\Gamma(\left\vert n-\xi\right\vert+1)}&=\sum_{n\leq\xi} \frac{\left(\frac{ M r\, \rho }{ 4\hbar t}\right)^{\xi-n}}{\Gamma( \xi-n+1)}+\sum_{n>\xi} \frac{\left(\frac{ M r\, \rho }{ 4\hbar t}\right)^{n-\xi}}{\Gamma(n-\xi+1)}
    \\
    &=\left(\frac{ M r\, \rho }{ 4\hbar t}\right)^{\xi_f}\sum_{k=0}^{\infty} \frac{\left(\frac{ M r\, \rho }{ 4\hbar t}\right)^k}{\Gamma(\xi_f+k+1)}+\left(\frac{ M r\, \rho }{ 4\hbar t}\right)^{-\xi_f}\sum_{k=1}^{\infty} \frac{\left(\frac{ M r\, \rho }{ 4\hbar t}\right)^{k}}{\Gamma(k+1-\xi_f)}
\end{split}
\]
so that
        \[
\begin{split}
    \sum_{n=-\infty}^{\infty} \frac{\left(\frac{ M r\, \rho }{ 4\hbar t}\right)^{\vert n-\xi\vert}}{\Gamma(\left\vert n-\xi\right\vert+1)}
    &\leq \left(\frac{ M r\, \rho }{ 4\hbar t}\right)^{\xi_f}\sum_{k=0}^{\infty} \frac{\left(\frac{ M r\, \rho }{ 4\hbar t}\right)^k}{k!}+\left(\frac{ M r\, \rho }{ 4\hbar t}\right)^{-\xi_f}\sum_{k=1}^{\infty} \frac{\left(\frac{ M r\, \rho }{ 4\hbar t}\right)^{k}}{(k-1)!}\\
    &= e^{\frac{ M r\, \rho }{ 4\hbar t}}\left[\left(\frac{ M r\, \rho }{ 4\hbar t}\right)^{\xi_f}+\left(\frac{ M r\, \rho }{ 4\hbar t}\right)^{1-\xi_f}\right].
\end{split}
\]
Collecting these estimates we finally obtain
\begin{align*}
    \vert F_{\xi}(r,\phi,\theta,t,\rho)\vert &\leq e^{\frac{ M r\, \rho }{ 4\hbar t}}I_0\left(\frac{Mr\rho}{2\hbar t}\right)\left\{\left(\frac{ M r\, \rho }{ 4\hbar t}\right)^{1-\xi_f}\right.\left.\left(2+\frac{ M r\, \rho }{ 4\hbar t}\right)+2\left(\frac{ M r\, \rho }{ 4\hbar t}\right)^{\xi_f}\left(1+\frac{ M r\, \rho }{ 4\hbar t}\right)\right.\\
    &\qquad\qquad\qquad\qquad\qquad\qquad\qquad \qquad\left.+\left[\left(\frac{ M r\, \rho }{ 4\hbar t}\right)^{\xi_f}+\left(\frac{ M r\, \rho }{ 4\hbar t}\right)^{1-\xi_f}\right]\right\}\\
    &\leq e^{\frac{ M r\, \rho }{ 4\hbar t}}I_0\left(\frac{Mr\rho}{2\hbar t}\right)\left\{\left(\frac{ M r\, \rho }{ 4\hbar t}\right)^{1-\xi_f}\left(3+\frac{Mr\rho}{4\hbar t}\right)+\left(\frac{ M r\, \rho }{ 4\hbar t}\right)^{\xi_f}\left(3+\frac{Mr\rho}{2\hbar t}\right)\right\}.
\end{align*}
This gives the statement
 by recalling the known, alternate expression, of the Bessel function, \cite[p. 181]{Watson}, we have
 $$
 I_\alpha(x)=\frac{1}{\pi}\int_0^\pi e^{x\cos\theta}\cos\theta d\theta-
 \frac{\sin(\alpha\pi)}{\pi}\int_0^\infty e^{-x\cosh(t)-\alpha t }dt,\quad \operatorname{Re}(x)>0
 $$
 from which we deduce that there exists a constant $K>0$ such that $|I_0(x)|\leq Ke^{x}$ for $x>0$.
\end{proof}

We want to use the Green's function $K(x'',x',t)$ to evolve the data $\psi_{0}(r,\phi)$.  This involves an integral that is difficult to compute and whose convergence is, a priori, in question.  However, using Fresnel integrals (changing the integral to one that is equivalent via an argument in complex analysis), we arrive at an expression for the evolution of the initial data that is easier to work with.  This is the content of the next proposition.

\begin{proposition}[Fresnel integral representation of the propagator]\label{PropPsiA}
Let $a,b\in \mathbb{R}$ and set
$$
\psi_0(r, \phi):=e^{ia x+iby}=e^{ia r\cos \phi +ibr\sin \phi}.
$$
Then the solution   
$$
\psi_{a,b}(r,\phi,t)
=\frac{M}{2\pi \hbar t} \int_{0}^{2\pi}\int_{0}^{\infty}  e^{\frac{iM}{2 \hbar t} (\rho^2+r^2)} F_{\xi}(r,\phi,\theta,t,\rho) e^{ia \rho\cos \theta +ib\rho\sin \theta}\ \rho d\rho d\theta
$$
can be written as a Fresnel integral as
\begin{equation}\label{solutFRESNEL}
\psi_{a,b}(r,\phi,t)=
i\frac{M}{2\pi \hbar t}e^{\frac{iM}{2 \hbar t} r^2}
\int_0^{2\pi}\int_0^\infty F_{\xi}(r,\phi,\theta,t,ue^{i\frac{\pi}{4}}) e^{\frac{-M}{2 \hbar t} u^2 } 
e^{i ue^{i\pi/4}(a\cos \theta +b\sin \theta)} u du d\theta.
\end{equation}
\end{proposition}
\begin{proof}
 We consider 
\begin{align*}
\psi_{a,b}(r,\phi,t)
& = \frac{M}{2\pi \hbar t} \int_{0}^{2\pi}\int_{0}^{\infty} e^{\frac{iM}{2 \hbar t} (\rho^2+r^2)} F_{\xi}(r,\phi,\theta,t,\rho) e^{ia \rho\cos \theta +ib\rho\sin \theta}\ \rho d\rho  d\theta\\
& = \frac{M}{2\pi \hbar t}e^{\frac{iM}{2 \hbar t} r^2}\int_{0}^{2\pi}\int_{0}^{\infty}  e^{\frac{iM}{2 \hbar t} \rho^2} F_{\xi}(r,\phi,\theta,t,\rho) e^{ia \rho\cos \theta +ib\rho\sin \theta}\ \rho d\rho d\theta.
\end{align*}
Fix $\theta\in[0,2\pi]$.  We now interpret the integral in $\rho$ as a Fresnel integral with the substitution $\rho\to z\in \mathbb{C}$.  Define
$$
g(z):=
   e^{\frac{iM}{2 \hbar t} z^2} F_{\xi}(r,\phi,\theta,t,z)e^{ia z\cos \theta +ibz\sin \theta}\ z.
$$
The function $g(z)$ is holomorphic (and entire) and so for any closed curve $\gamma$ we have:
$$
\int_\gamma g(z)dz=
\int_\gamma  e^{\frac{iM}{2 \hbar t} z^2} F_{\xi}(r,\phi,\theta,t,z)e^{i z(a\cos \theta +b\sin \theta)}\ z dz=0.
$$
We proceed by considering the path of integration
 $$
 \gamma=\gamma_1\cup\gamma_2\cup\gamma_3
 $$ where, for $z=u e^{i\alpha}$, and $R>0$ we set
$$
\gamma_1=\{ 0\leq u\leq R,\ \ \alpha =0\}, \ \ \ \
\gamma_2=\{ u= R,\ \ 0\leq \alpha \leq \pi/4\},
\ \ \ \  -\gamma_3=\{ 0\leq u\leq R,\ \ \alpha =\pi/4\},
$$
where the sign in front of $\gamma_3$ takes into account the orientation.
We will exploit the fact that
$$
\int_\gamma g(z)dz=0
$$
to change the integral we are after into one that has Gaussian decay and is in the conclusion of the Proposition.

Firstly, let us consider the integral on $\gamma_1$.  Here the parameterization of $\gamma_1$ is given by $z=u$ and so $dz=du$ and then using this we see that $I_1(R)$ becomes
\begin{align*}
I_{1}(R)&:=\int_{\gamma_1}g(z)dz
\\
&
=\int_{0}^R e^{\frac{iM}{2 \hbar t} u^2} F_{\xi}(r,\phi,\theta,t,ue^{i\alpha})
e^{i u(a\cos \theta +b\sin \theta)}\ u du.
\end{align*}

We next consider the integral over the curve $\gamma_2$.  Here we have the parameterization $z=R e^{i\alpha}$, which gives $dz=Ri e^{i\alpha} d \alpha$ and gives the following for $I_2(R)$,
\begin{align*}
I_{2}(R)&:=\int_{\gamma_2}g(z)dz\\
&=iR^2\int_{0}^{\pi/4} e^{\frac{iM}{2 \hbar t} R^2e^{2i\alpha}}F_{\xi}(r,\phi,\theta,t,Re^{i\alpha})
e^{i Re^{i\alpha}(a\cos \theta +b\sin \theta)}e^{2i\alpha}d\alpha\\
&=iR^2\int_{0}^{\pi/4} e^{\frac{M}{2 \hbar t} R^2(i\cos(2\alpha)-\sin(2\alpha))}F_{\xi}(r,\phi,\theta,t,Re^{i\alpha})
e^{i Re^{i\alpha}(a\cos \theta +b\sin \theta)}e^{2i\alpha}d\alpha.
\end{align*}

\noindent Hence, 

\begin{align*}
\vert I_{2}(R)\vert&:=R^2\int_{0}^{\pi/4} e^{-\frac{M}{2 \hbar t} R^2\sin(2\alpha)}\left\vert F_{\xi}(r,\phi,\theta,t,Re^{i\alpha})\right\vert
e^{-R\sin(\alpha)(a\cos \theta +b\sin \theta)}d\alpha\\
&\leq R^2e^{-\frac{M}{2 \hbar t} R^2}\int_{0}^{\pi/4}\left\vert F_{\xi}(r,\phi,\theta,t,Re^{i\alpha})\right\vert e^{-R\sin(\alpha)(a\cos\theta +b\sin\theta)}d\alpha.
\end{align*}
The second integral goes to zero because of the estimates on $\left\vert F_{\xi}(r,\phi,\theta,t,Re^{i\alpha})\right\vert$ from \eqref{e:F_Estimates} which shows that this integral can be controlled by $q(R)e^{-cR^2}e^{\widetilde{c}R}$ for some function $q$ with at most polynomial growth and constants $c$ and $\widetilde{c}$, but this clearly goes to $0$ as $R\to\infty$.

Finally, we turn to the integral over the path $\gamma_3$.  Here we have that $z=ue^{i\frac{\pi}{4}}$ and so $dz=e^{i\frac{\pi}{4}}du$ and then using this parameterization we see that 
\begin{align*}
    I_{3}(R)&:=\int_{\gamma_3}g(z)dz=-\int_0^{R}g(ue^{i\frac{\pi}{4}})e^{i\frac{\pi}{4}} du
\\
&= -e^{i\frac{\pi}{2}}\int_{0}^{R} e^{\frac{iM}{2 \hbar t} u^2 e^{i \frac{\pi}{2}}} F_{\xi}(r,\phi,\theta,t,ue^{i\frac{\pi}{4}})
e^{i ue^{i \frac{\pi}{4}}(a\cos \theta +b\sin \theta)}\ u du 
\\
&
= -i\int_{0}^{R} e^{\frac{iM}{2 \hbar t} u^2 i} F_{\xi}(r,\phi,\theta,t,ue^{i\frac{\pi}{4}})
e^{i ue^{i\frac{\pi}{4}}(a\cos \theta +b\sin \theta)}\ u du
\\
&
= -i\int_{0}^{R} e^{\frac{-M}{2 \hbar t} u^2 }F_{\xi}(r,\phi,\theta,t,ue^{i\frac{\pi}{4}})
e^{i ue^{i\frac{\pi}{4}}(a\cos \theta +b\sin \theta)} u du.
\end{align*}
Using these integrals, we have that 
$$
\int_{\gamma}g(z)dz=\sum_{\ell=1}^3I_{\ell}(R)=0
$$
and upon taking the limit as $R\to+\infty$ we
have
$$
\lim_{R\to + \infty}I_{1}(R)=-\lim_{R\to +\infty}I_{3}(R)
$$
since $\displaystyle\lim_{R\to + \infty}I_{2}(R)=0$.  This then gives that
\begin{align}
\lim_{R\to +\infty}I_{1}(R) & = \int_{0}^\infty e^{\frac{iM}{2 \hbar t} u^2} F_{\xi}(r,\phi,\theta,t,ue^{i\alpha}) e^{i u(a\cos \theta +b\sin \theta)}\ u du \notag\\ 
&=-\lim_{R\to +\infty}I_{3}(R)
=i\int_{0}^{\infty} e^{\frac{-M}{2 \hbar t} u^2 } F_{\xi}(r,\phi,\theta,t,ue^{i\frac{\pi}{4}})
e^{i ue^{i\frac{\pi}{4}}(a\cos \theta +b\sin \theta)}  u du. \label{FREN}
\end{align}
We can now use \eqref{FREN} and integrate the resulting identity in $\theta$ to replace the integral:
$$
\psi_{a,b}(r,\phi,t):=\frac{M}{2\pi \hbar t}e^{\frac{iM}{2 \hbar t} r^2}\int_{0}^{2\pi}\int_{0}^{\infty}  e^{\frac{iM}{2 \hbar t} \rho^2} F_{\xi}(r,\phi,\theta,t,\rho) e^{ia \rho\cos \theta +ib\rho\sin \theta}\ \rho d\rho d\theta
$$
with the Fresnel integral, given by
$$
\psi_{a,b}(r,\phi,t)=
i\frac{M}{2\pi \hbar t}e^{\frac{iM}{2 \hbar t} r^2}
\int_0^{2\pi}\int_0^\infty F_{\xi}(r,\phi,\theta,t,ue^{i\frac{\pi}{4}}) e^{\frac{-M}{2 \hbar t} u^2 } 
e^{i ue^{i\frac{\pi}{4}}(a\cos \theta +b\sin \theta)} u du d\theta
$$
giving the statement of the Proposition.
\end{proof}

\section{Aharonov-Bohm Infinite Order Differential Operators and \\ the Supershift Property}
\label{s:ProofMain}

Based on the Fresnel integral representation of the solution of the Schr\"odinger equation,
given in Proposition \ref{PropPsiA}, see (\ref{solutFRESNEL}), we further have to manipulate this  solution
to represent it as a power series of the coefficients $a$ and $b$ appearing in the superoscillatory function associated with the initial datum.

\begin{lemma}[Series expansion of the solution in the parameters $a$ and $b$]
\label{l:Series}
Let $\psi_{a,b}(r,\phi,t)$ be the solution of the Schr\"odinger equation given by   \eqref{solutFRESNEL}. Then we have
\begin{align*}
&\psi_{a,b}(r,\phi,t)=
i\frac{M}{2\pi \hbar t}e^{\frac{iM}{2 \hbar t}r^2}\sum_{m=0}^{\infty} \frac{\left(i e^{i\pi/4}\right)^m}{m!}
\sum_{\ell=0}^m\binom{m}{\ell}a^{m-\ell}b^\ell c_{m,\ell}(r,\phi,t),
 \end{align*}
where
\begin{equation}\label{coeffcml}
c_{m,\ell}(r,\phi,t):=\int_0^{2\pi}\int_0^\infty
e^{\frac{-M}{2 \hbar t}u^2} F_{\xi}(r,\phi,\theta,t,ue^{i\frac{\pi}{4}})
 \cos^{m-\ell}\theta\sin^{\ell}\theta u^{m+1} du
  d\theta.
 \end{equation}

\end{lemma}
\begin{proof}
    
We consider 
$$
\psi_{a,b}(r,\phi,t)=
i\frac{M}{2\pi \hbar t}e^{\frac{iM}{2 \hbar t} r^2}
\int_0^{2\pi}\int_0^\infty F_{\xi}(r,\phi,\theta,t,ue^{i\frac{\pi}{4}}) e^{\frac{-M}{2 \hbar t} u^2 }
e^{i ue^{i\pi/4}(a\cos \theta +b\sin \theta)} u du d\theta
$$
and we expand the exponential $e^{i ue^{i\pi/4}(a\cos \theta +b\sin \theta)}$ to get

$$\psi_{a,b}(r,\phi,t)=i\frac{M}{2\pi \hbar t}
e^{\frac{iM}{2 \hbar t}r^2}
\int_0^{2\pi}\int_0^\infty e^{\frac{-M}{2 \hbar t} u^2 }  F_{\xi}(r,\phi,\theta,t,ue^{i\frac{\pi}{4}}) \sum_{m=0}^{\infty}
\frac{\left(i ue^{i\pi/4}(a\cos \theta+b\sin \theta)\right)^m}{m!}u du
  d\theta.$$
 We want to exchange the infinite sum and the integrals involved.  This will be accomplished with the Lebesgue Dominated Convergence Theorem and the estimates on the function $F_{\xi}(r,\phi,\theta,t,\rho)$ from \eqref{e:F_Estimates}.  To do this, for $L\in\mathbb{N}$, define
 $$
 \displaystyle G_{\xi,L}(r,\phi,\theta,t,ue^{i\frac{\pi}{4}}):=e^{\frac{-M}{2 \hbar t}u^2}F_{\xi}(r,\phi,\theta,t,ue^{i\frac{\pi}{4}}) \sum_{m=0}^L
\frac{\left(i ue^{i\pi/4}(a\cos \theta+b\sin \theta)\right)^m}{m!}u.$$
By the computations above, in particular those connected to proving \eqref{e:F_Estimates}, one can show that for $u$ large and for all $\theta\in[0,2\pi)$:
\begin{align*}
 \vert G_{\xi,L}(r,\phi,\theta,t,ue^{i\frac{\pi}{4}})\vert&\lesssim e^{\frac{-M }{2 \hbar t}u^2} e^{\tilde{c}u} q(u) \sum_{m=0}^L\frac{u^{m}(\vert a\vert+\vert b\vert)^m}{m!}\\
 &\lesssim e^{\frac{-M}{2 \hbar t}u^2} e^{\tilde{c}u} q(u) e^{u(\vert a\vert+\vert b\vert)}=e^{\frac{-M }{2 \hbar t}u^2} e^{\tilde{c}u} q(u),
\end{align*}
for some constant $\tilde{c}$ and for some function $q$ with at most polynomial growth in $u$.  This function is then clearly integrable on $[0,2\pi]\times [0,\infty)$. The Lebesgue Dominated Convergence  Theorem then allows for the  interchange of the sum and the double integral, so the function $\psi_{a,b}=\psi_{a,b}(r,\phi,t)$, turns out to be 
\begin{align*}
\psi_{a,b} &=i\frac{M}{2\pi \hbar t}
e^{\frac{iM}{2 \hbar t}r^2}\sum_{m=0}^{\infty} \frac{\left(i e^{i\pi/4}\right)^m}{m!}
\int_0^{2\pi}\int_0^\infty
e^{\frac{-M }{2 \hbar t}u^2} F_{\xi}(r,\phi,\theta,t,ue^{i\frac{\pi}{4}})
\left(a\cos \theta+b\sin\theta\right)^m u^{m+1} du
d\theta.
 \end{align*}

Next, expand the binomial $\displaystyle\left(a\cos \theta+b\sin\theta\right)^m =\sum_{\ell=0}^m\binom{m}{\ell}a^{m-\ell}b^\ell \cos^{m-\ell}\theta\sin^{\ell}\theta
$
and obtain
\begin{align*}
\psi_{a,b}(r,\phi,t) &=i\frac{M}{2\pi \hbar t}
e^{\frac{iM}{2 \hbar t}r^2}\sum_{m=0}^{\infty} \frac{\left(i e^{i\pi/4}\right)^m}{m!}
\sum_{\ell=0}^m\binom{m}{\ell}a^{m-\ell}b^\ell c_{m,\ell}(r,\phi,t),
 \end{align*}
where
\begin{equation*}
c_{m,\ell}(r,\phi,t):=\int_0^{2\pi}\int_0^\infty
e^{\frac{-M}{2 \hbar t}u^2} F_{\xi}(r,\phi,\theta,t,ue^{i\frac{\pi}{4}})
 \cos^{m-\ell}\theta\sin^{\ell}\theta u^{m+1} du
  d\theta
 \end{equation*}
 as claimed in \eqref{coeffcml} and completing the proof of the Lemma.
  \end{proof}

\begin{remark}\label{RKsuporfunc}
The superoscillatory function $Y_n(x,y,a)$ defined in (\ref{YN}),
with coefficients $C_j(n,a)$ in  (\ref{CICONJN}),
contains the entire functions $g$ and $h$,  monotone increasing in $a$, with
  \begin{equation}\label{funzGH}
  g(\lambda)=\sum_{u=0}^\infty g_u\lambda^u, \ \ \ h(\lambda)=\sum_{v=0}^\infty h_v\lambda^v.
  \end{equation}
 This function satisfies
\begin{equation}\label{YNNN}
\lim_{n\to\infty}Y_n(x,y,a)=e^{ig(a)x}e^{ih(a)y}
\end{equation}
pointwise in $x$ and $y$ on compact sets in $\mathbb{R}^2$.
The important fact in our considerations is that the function 
$F_n(x,a)$ defined in (\ref{FNEXP}), extended to an entire complex function, converges in the space $A_1$.
\end{remark}

The crucial problem is now to represent the solution of the Schr\"odinger equation with the superoscillatory initial datum given in Remark \ref{RKsuporfunc} via an infinite order differential operators.
The series expansions of the entire functions $g$ and $h$ given in 
\eqref{funzGH} will be now used to define the 
infinite order differential operator associated with the Aharonov-Bohm field.

\begin{definition}[Infinite order differential operator for the Aharonov-Bohm field]\label{OPER-AHA-BOOOM}
Let $w\in \mathbb{C}$ be an auxiliary variable.
Denote by $\mathcal{D}_{w}$ the complex derivative with respect to the variable $w$ and
define the formal infinite order differential operators, for
$m\geq\ell$,  where $m,\ell\in \mathbb{N}_0$ as 
\begin{equation}\label{operG}
\mathcal{G}_{m,\ell}(\mathcal{D}_w):=\left(\sum_{u=0}^\infty \frac{g_u}{i^u}\mathcal{D}_w^u\right)^{m-\ell},
\end{equation}
and
\begin{equation}\label{operH}
\mathcal{H}_{\ell}(\mathcal{D}_w):=\left(\sum_{v=0}^\infty \frac{h_v}{i^v}\mathcal{D}_w^v\right)^\ell,
\end{equation}
where the coefficients $g_u$ and $h_v$ are given by \eqref{funzGH}.
We define the infinite order differential operator associated with the Aharonov-Bohm field
as
\begin{equation}\label{OPERU}
\mathcal{U}(r,\phi,t;\mathcal{D}_w):=i\frac{M}{2\pi \hbar t}
e^{\frac{iM}{2 \hbar t}r^2}\sum_{m=0}^{\infty} \frac{\left(i e^{i\pi/4}\right)^m}{m!}
\sum_{\ell=0}^m\binom{m}{\ell} c_{m,\ell}(r,\phi,t)
\mathcal{G}_{m,\ell}(\mathcal{D}_w) \mathcal{H}_{\ell}(\mathcal{D}_w),
\end{equation}
where $\mathcal{H}_{\ell}(\mathcal{D}_w)$ 
and $ \mathcal{H}_{\ell}(\mathcal{D}_w)$ are
defined in \eqref{operG} and \eqref{operH} respectively.
\end{definition}

\begin{theorem}[Representation of the solution via the Aharonov-Bohm infinite order differential operator] \label{INFTYORDESSOLUT}
The  solution of the Schr\"odinger equation in the Aharonov-Bohm field with superoscillatory initial datum as in Remark \ref{RKsuporfunc} can be written 
 in terms
of the infinite order differential operators $\mathcal{G}_{m,\ell}(\mathcal{D}_w)$ and $\mathcal{H}_{\ell}(\mathcal{D}_w)$  as
\begin{equation}
\begin{split}
\psi_{g(a),h(a)}(r,\phi,t) &=\mathcal{U}(r,\phi,t;\mathcal{D}_\xi)e^{iaw}\Big|_{w=0}
\\
&
=
i\frac{M}{2\pi \hbar t}e^{\frac{iM}{2 \hbar t}r^2}\sum_{m=0}^{\infty} \frac{\left(i e^{i\pi/4}\right)^m}{m!}
\sum_{\ell=0}^m\binom{m}{\ell} c_{m,\ell}(r,\phi,t)
\mathcal{G}_{m,\ell}(\mathcal{D}_w) \mathcal{H}_{\ell}(\mathcal{D}_w)
e^{iaw}\Big|_{w=0}.\label{e:RepInfDO}
\end{split}
\end{equation}
\end{theorem}
\begin{proof}
In the function $\psi_{a,b}(r,\phi,t)$, given in  Lemma \ref{l:Series},
we replace the parameter $a$ by the entire function 
$g(a)$ and $b$ by the entire function $h(a)$
so that we have
 \begin{align}
 \label{psigh}
\psi_{g(a),h(a)}(r,\phi,t) &=i\frac{M}{2\pi \hbar t}
e^{\frac{iM}{2 \hbar t}r^2}\sum_{m=0}^{\infty} \frac{\left(i e^{i\pi/4}\right)^m}{m!}
\sum_{\ell=0}^m\binom{m} {\ell}c_{m,\ell}(r,\phi,t) (g(a))^{m-\ell}(h(a))^\ell .
 \end{align}
 Recalling the series expansions in (\ref{funzGH}) for the functions $g(a)$ and $h(a)$ and substituting this into \eqref{psigh} gives
 \begin{align*}
\psi_{g(a), h(a)}(r,\phi,t) &=i\frac{M}{2\pi \hbar t}
e^{\frac{iM}{2 \hbar t}r^2}\sum_{m=0}^{\infty} \frac{\left(i e^{i\pi/4}\right)^m}{m!}
\sum_{\ell=0}^m\binom{m}{\ell} c_{m,\ell}(r,\phi,t)\left(\sum_{u=0}^\infty g_u a^u\right)^{m-\ell}\left(\sum_{v=0}^\infty h_v a^v\right)^\ell .
 \end{align*}

Observe that using the auxiliary complex variable $w$ we have
\begin{equation*}
\lambda^\ell=\frac{1}{i^\ell}\mathcal{D}_w^\ell e^{iw \lambda}\Big|_{w=0}\ \ \ {\rm for }\ \ \ \lambda\in \mathbb{C}, \ \ \ \ell\in \mathbb{N},
\end{equation*}
where $\mathcal{D}_w$ is the derivative with respect to $w$ and $|_{w=0}$ denotes evaluation at $w=0$.   This coupled with the computations above gives \eqref{e:RepInfDO}, and proves the statement of the theorem.
\end{proof}

\begin{theorem}\label{CONTI}
Let $A_1$ be the space of entire functions as in Definition \ref{sect2-def1}.
Then the infinite order differential operator for the Aharonov-Bohm field $\mathcal{U}(r,\phi,t;\mathcal{D}_w)$, defined in \eqref{operH},
acts continuously from $A_1$ to $A_1$.
\end{theorem}

\begin{proof}
We split the proof in steps.

\medskip
\noindent
{\em Step 1: Estimates for the operators $\mathcal{G}_{m,\ell}(\mathcal{D}_w)$ and $ \mathcal{H}_{\ell}(\mathcal{D}_w)$.}
For $m\geq\ell$ for $m,\ell\in \mathbb{N}$, using the property
$$
\mathcal{D}_w^{v}\sum_{j=0}^\infty f_jw^j
=\sum_{j=v}^\infty f_j \frac{j!}{(j-v)!}w^{j-v}
=\sum_{k=0}^\infty f_{v+k} \frac{(v+k)!}{k!}w^{k}
$$
and iterating, $\ell$ times, using the relation
$$
\left(\sum_{v_1=0}^\infty \frac{h_{v_1}}{i^{v_1}}\mathcal{D}_w^{v_1}\right) \sum_{j=0}^\infty f_jw^j=
\sum_{v_1=0}^\infty \frac{h_{v_1}}{i^{v_1}} \sum_{k=0}^\infty f_{v_1+k}\frac{(v_1+k)!}{k!} w^k
$$
we have

\begin{align*}
     \mathcal{H}_{\ell}(\mathcal{D}_w)f(w)&=
\left(\sum_{v=0}^\infty \frac{h_v}{i^v}\mathcal{D}_w^v\right)^\ell f(w)
\\
&
=
\sum_{v_1=0}^\infty \frac{h_{v_1}}{i^{v_1}}\ldots \sum_{v_\ell=0}^\infty \frac{h_{v_\ell}}{i^{v_\ell}}
 \sum_{k=0}^\infty f_{\sum_{j=1}^{\ell} v_j+k}\frac{(\sum_{j=1}^\ell v_j+k)!}{k!} w^k
 \\
 &
 =
\prod_{j=1}^\ell\sum_{v_j=0}^\infty \frac{h_{v_j}}{i^{v_j}}
 \sum_{k=0}^\infty f_{\sum_{j=1}^\ell v_j+k}\frac{(\sum_{j=1}^\ell v_j+k)!}{k!} w^k.
\end{align*}

Similarly considering the operators
$\mathcal{G}_{m,\ell}(\mathcal{D}_w)$, and performing analogous computations we have
\[
\begin{split}
&\mathcal{G}_{m,\ell}(\mathcal{D}_w) \mathcal{H}_{\ell}(\mathcal{D}_w)f(w)=
\left(\sum_{u=0}^\infty \frac{g_u}{i^u}\mathcal{D}_w^u\right)^{m-\ell}\left(\sum_{v=0}^\infty \frac{h_v}{i^v}\mathcal{D}_w^v\right)^\ell f(w)
\\
&
=\prod_{p=1}^{m-\ell}\sum_{u_p=0}^\infty  \frac{g_{u_p}}{i^{u_p}}   \prod_{j=1}^\ell\sum_{v_j=0}^\infty \frac{h_{v_j}}{i^{v_j}}
 \sum_{k=0}^\infty f_{\sum_{p=1}^{m-\ell} u_p+\sum_{j=1}^\ell v_j+k}\frac{\left(\sum_{p=1}^{m-\ell} u_p+\sum_{j=1}^\ell v_j+k\right)!}{k!} w^k.
\end{split}
\]

We now estimate $\mathcal{G}_{m,\ell}(\mathcal{D}_w) \mathcal{H}_{\ell}(\mathcal{D}_w)f(w)$ taking the modulus
\[
\begin{split}
&\Big|\mathcal{G}_{m,\ell}(\mathcal{D}_w) \mathcal{H}_{\ell}(\mathcal{D}_w)f(w)\Big|
\\
&
\leq \prod_{p=1}^{m-\ell}\sum_{u_p=0}^\infty |g_{u_p}|   \prod_{j=1}^\ell\sum_{v_j=0}^\infty |h_{v_j}|
 \sum_{k=0}^\infty 
 \big|f_{\sum_{p=1}^{m-\ell} u_p+\sum_{j=1}^\ell v_j+k}\big|\frac{\left(\sum_{p=1}^{m-\ell} u_p+\sum_{j=1}^\ell v_j+k\right)!}{k!}\ | w|^k
\end{split}
\]
and using the estimate in Lemma \ref{BETOA1} for the coefficients 
$f_{\sum_{p=1}^{m-\ell} u_p+\sum_{j=1}^\ell v_j+k}$ we have for some $C_f>0$ and $b>0$ that 
$$
\left|f_{\sum_{p=1}^{m-\ell} u_p+\sum_{j=1}^\ell v_j+k}\right|
\leq C_f\frac{b^{\sum_{p=1}^{m-\ell} u_p+\sum_{j=1}^\ell v_j+k}}{\Gamma\left(\sum_{p=1}^{m-\ell} u_p+\sum_{j=1}^\ell v_j+k+1\right)}.
$$
By elementary properties of the Gamma function we have
\[
\begin{split}
&\Big|\mathcal{G}_{m,\ell}(\mathcal{D}_w) \mathcal{H}_{\ell}(\mathcal{D}_w)f(w)\Big|
\\
&
\leq C_f\prod_{p=1}^{m-\ell}\sum_{u_p=0}^\infty  |g_{u_p}|
  \prod_{j=1}^\ell\sum_{v_j=0}^\infty |h_{v_j}|
 \sum_{k=0}^\infty \frac{\,b^{\sum_{p=1}^{m-\ell} u_p+\sum_{j=1}^\ell v_j+k} \,\left(\sum_{p=1}^{m-\ell} u_p+\sum_{j=1}^\ell v_j+k\right)!}{\Gamma\left(\sum_{p=1}^{m-\ell} u_p+\sum_{j=1}^\ell v_j+k+1\right)\,k!}
 |w|^k.
\end{split}
\]
Simplifying we have
\begin{align}
&\Big|\mathcal{G}_{m,\ell}(\mathcal{D}_w) \mathcal{H}_{\ell}(\mathcal{D}_w)f(w)\Big|
\notag\\
&
\leq C_f\prod_{p=1}^{m-\ell}\sum_{u_p=0}^\infty  |g_{u_p}|b^{\sum_{p=1}^{m-\ell} u_p}
  \prod_{j=1}^\ell\sum_{v_j=0}^\infty |h_{v_j}|b^{\sum_{j=1}^\ell v_j}\
 \sum_{k=0}^\infty \frac{b^k}{k!} \ |w|^k
 \notag\\
&
= C_f\left(\sum_{u=0}^\infty  |g_{u}|b^u\right)^{m-\ell} 
  \left(\sum_{v=0}^\infty |h_{v}|b^v\right)^\ell\
e^{b |w|}
\label{EstimateIODO}
 \end{align}
and we observe that since the function $g$ and $h$ are entire 
$$
\sum_{u=0}^\infty  |g_{u}|b^u <\infty \ \ {\rm and} \ \ \sum_{v=0}^\infty  |h_{v}|b^v<\infty
$$
for all $b\in \mathbb{R}$.

\medskip
\noindent
{\em Step 2: Estimates for the coefficients $c_{m,\ell}(r,\phi,t)$.}
We consider the modulus of the coefficients in \eqref{coeffcml}
\[
\begin{split}
|c_{m,\ell}(r,\phi,t)|&=
\left|\int_0^{2\pi}\int_0^\infty
e^{\frac{-M}{2 \hbar t}u^2} F_{\xi}(r,\phi,\theta,t,ue^{i\frac{\pi}{4}})
 \cos^{m-\ell}\theta\sin^{\ell}\theta u^{m+1} du
 \ d\theta\right|
\\
&
\leq
\int_0^{2\pi}\int_0^\infty
e^{\frac{-M}{2 \hbar  t}u^2} \big|F_{\xi}(r,\phi,\theta,t,ue^{i\frac{\pi}{4}})
 \cos^{m-\ell}\theta\sin^{\ell}\theta u^{m+1}\big| du
 \\
 &\leq
\int_0^{2\pi}\int_0^\infty
e^{\frac{-M}{2 \hbar t}u^2} \big|F_{\xi}(r,\phi,\theta,t,ue^{i\frac{\pi}{4}})\big|
  u^{m+1} du
 d\theta.
 \end{split}
 \]
Now, thanks to Lemma \ref{EFFEXI}
for $u$ real and non-negative there exists a positive constant $c$ such that
\begin{equation}
    \left\vert F_{\xi}(r,\phi,\theta,t,ue^{i\frac{\pi}{4}})\right\vert \lesssim q(u)e^{{c}u}
\end{equation}
where $q(u)$ is a function that has at most polynomial growth in $u$ and the coefficients in the function and in
${c}$ depend upon $r,M,t,\theta,\phi$ and $\xi$.
More explicitly, 
\begin{align*}
\left\vert F_{\xi}(r,\phi,\theta,t,\rho)\right\vert
  \leq e^{\frac{ M r\, \rho }{ 4\hbar t}}I_0\left(\frac{Mr\rho}{2\hbar t}\right)\left\{\left(\frac{ M r\, \rho }{ 4\hbar t}\right)^{1-\xi_f}\left(3+\frac{Mr\rho}{4\hbar t}\right)+\left(\frac{ M r\, \rho }{ 4\hbar t}\right)^{\xi_f}\left(3+\frac{Mr\rho}{2\hbar t}\right)\right\},
\end{align*}
where $\rho=|ue^{i\frac{\pi}{4}}|=u$. So, for $t>0$, we have
\[
\begin{split}
|c_{m,\ell}(r,\phi,t)|&=
\int_0^{2\pi}\int_0^\infty
e^{\frac{-M}{2 \hbar t}u^2} \big|F_{\xi}(r,\phi,\theta,t,ue^{i\frac{\pi}{4}})\big|
  u^{m+1} du
 \ d\theta
\\
&
\leq\int_0^{2\pi}\int_0^\infty
e^{\frac{-M}{2 \hbar t}u^2} q(u)e^{{c}u}  u^{m+1} du \ d\theta
\\
&
\leq 2\pi\int_0^\infty e^{\frac{-M}{4 \hbar t}u^2} \Big(e^{\frac{-M}{4 \hbar t}u^2} q(u)e^{{c}u}\Big)u^{m+1} du
\\
&
\leq K \int_0^\infty e^{\frac{-M}{4 \hbar  t}u^2}  u^{m+1} du
 \end{split}
 \]
 where we set
 $
 \displaystyle K:=2\pi \sup_{u\in [0,\infty]}\Big(e^{\frac{-M}{4 \hbar t}u^2} q(u)e^{{c}u}\Big)<\infty
$.
Using the following well known integral
\begin{equation*}
\int_{0}^{\infty} x^n e^{-\gamma x^2}dx=\frac{\Gamma\Big(\frac{n+1}{2}\Big)}{2\gamma^{(n+1)/2}},\ \ \ {\rm for}\ \  \ \gamma>0.
\end{equation*}
we conclude that
\begin{equation}\label{STIMACML}
|c_{m,\ell}(r,\phi,t)|\leq K  \frac{\Gamma\Big(\frac{m+2}{2}\Big)}{2\Big(\frac{M}{4 \hbar t}\Big)^{(m+2)/2}}.
\end{equation}
Observe that the estimate does not depend on $\ell$.

\medskip
\noindent
{\em Step 3: The continuity estimate.}
Based on the previous points, for $t>0$, we now consider
\[
\begin{split}
&\left|\mathcal{U}(r,\phi,t;\mathcal{D}_w)f(w)\right|=
\left|i\frac{M}{2\pi \hbar t}
e^{\frac{iM}{2 \hbar t}r^2}\sum_{m=0}^{\infty} \frac{\left(i e^{i\pi/4}\right)^m}{m!}
\sum_{\ell=0}^m\binom{m}{\ell} c_{m,\ell}(r,\phi,t)
\mathcal{G}_{m,\ell}(\mathcal{D}_w) \mathcal{H}_{\ell}(\mathcal{D}_w)f(w)\right|
\\
&
\leq \frac{M}{2\pi \hbar t}
\sum_{m=0}^{\infty} \frac{1}{m!}
\sum_{\ell=0}^m\binom{m}{\ell} \big| c_{m,\ell}(r,\phi,t)\big|\
\big| \mathcal{G}_{m,\ell}(\mathcal{D}_w) \mathcal{H}_{\ell}(\mathcal{D}_w)f(w)\big|
\\
&
\leq \frac{M}{2\pi \hbar t}
\sum_{m=0}^{\infty} \frac{1}{m!}
\sum_{\ell=0}^m\binom{m}{\ell}  K  \frac{\Gamma\Big(\frac{m+2}{2}\Big)}{2\Big(\frac{M}{4 \hbar t}\Big)^{(m+2)/2}}
C_f\left(\sum_{u=0}^\infty  |g_{u}|b^u\right)^{m-\ell}
  \left(\sum_{v=0}^\infty |h_{v}|b^v\right)^\ell\
e^{b |w|}
\\
&
= \frac{M}{2\pi \hbar t}KC_f
\sum_{m=0}^{\infty} \frac{1}{m!} \frac{\Gamma\Big(\frac{m+2}{2}\Big)}{2\Big(\frac{M}{4 \hbar t}\Big)^{(m+2)/2}}
\sum_{\ell=0}^m\binom{m}{\ell}
\left(\sum_{u=0}^\infty  |g_{u}|b^u\right)^{m-\ell}
  \left(\sum_{v=0}^\infty |h_{v}|b^v\right)^\ell\
e^{b |w|}
\\
&
= \frac{M}{2\pi \hbar t}KC_f
\sum_{m=0}^{\infty} \frac{1}{m!} \frac{\Gamma\left(\frac{m+2}{2}\right)}{2\left(\frac{M}{4 \hbar t}\right)^{(m+2)/2}}
\left(\sum_{u=0}^\infty  |g_{u}|b^u+\sum_{v=0}^\infty |h_{v}|b^v\right)^m\ e^{b |w|},
\end{split}
\]
where we have utilized the estimates \eqref{STIMACML} and \eqref{EstimateIODO} in the course of the estimation above.
Now setting
$$
\Lambda:=\sum_{m=0}^{\infty} \frac{1}{m!} \frac{\Gamma\Big(\frac{m+2}{2}\Big)}{2\Big(\frac{M}{4 \hbar t}\Big)^{(m+2)/2}}
\left(\sum_{u=0}^\infty  |g_{u}|b^u+\sum_{v=0}^\infty |h_{v}|b^v\right)^m
$$
using the estimate in Lemma \ref{ESTM_GAMMA}, for $q=2$:
$$
\Gamma\left(\frac{m}{2}+1\right)\leq (m!)^{1/2}
$$
we have 
\[
\begin{split}
\Lambda&\leq 
\sum_{m=0}^{\infty} \frac{1}{m!} \frac{(m!)^{1/2}}{2\Big(\frac{M}{4 t}\Big)^{(m+2)/2}}
\left(\sum_{u=0}^\infty  |g_{u}|b^u+\sum_{v=0}^\infty |h_{v}|b^v\right)^m
\\
&
=
\frac{2t}{M}\sum_{m=0}^{\infty} \frac{1}{(m!)^{1/2}} 
\left( \left(\frac{4t}{M}\right)^{1/2} \left(\sum_{u=0}^\infty  |g_{u}|b^u+\sum_{v=0}^\infty |h_{v}|b^v\right)\right)^m<\infty.
\end{split}
\]
All together this gives the estimate
$$
\big|\mathcal{U}(r,\phi,t;\mathcal{D}_w)f(w)\big|\leq \frac{M}{2\pi \hbar t} KC_f\Lambda\ e^{b |w|}
$$
which gives us the statement of the theorem.  Note that the constants involved depend upon $M,t,r,\phi$ and $\xi$.
\end{proof}

Due to the continuity theorem proved above now we can show that solutions of the Schr\"odinger equation with superoscillatory datum under the Aharonov-Bohm field possess the supershift property.

 \begin{theorem}\label{MAINRES}
  The solution of the Schr\"odinger equation,
 with the propagator \eqref{PROPAB} with superoscillatory initial datum
 given by \eqref{YN} can be written as
$$
\Psi_n(r,\phi,t)=\sum_{j=0}^nC_j(n,a) \psi_{g\left(1-\frac{2j}{n}\right), h\left(1-\frac{2j}{n}\right)}(r,\phi,t)
$$
for all $r\in (0,\infty)$ $\phi \in [0,2\pi]$, $t >0$,
where $\psi_{g(a),h(a)}(r,\phi,t)$ is given by \eqref{e:RepInfDO}  and $C_j(n,a)$ are given by \eqref{CICONJN}.
Moreover,
$\psi_{a,b}(r,\phi,t)$
has the supershift property, that is
$$
\lim_{n\to\infty}\Psi_n(r,\phi,t)=\lim_{n\to\infty}\sum_{j=0}^nC_j(n,a)\psi_{g\left(1-\frac{2j}{n}\right), h\left(1-\frac{2j}{n}\right)}(r,\phi,t)=\psi_{g(a), h(a)}(r,\phi,t).
 $$
\end{theorem}
\begin{proof}
Since the solution 
$\psi_{g(a),h(a)}(r,\phi,t)$ can be 
represented via the Aharonov-Bohm infinite order differential operator, 
from Theorem \ref{INFTYORDESSOLUT},  
$\psi_{g\left(1-\frac{2j}{n}\right), h\left(1-\frac{2j}{n}\right)}(r,\phi,t)$ is given in terms
of the infinite order differential operators 
$\mathcal{G}_{m,\ell}(\mathcal{D}_w)$ and $\mathcal{H}_{\ell}(\mathcal{D}_w)$  as
\begin{equation*}
\begin{split}
&\psi_{g\left(1-\frac{2j}{n}\right), h\left(1-\frac{2j}{n}\right)}(r,\phi,t)
=\mathcal{U}(r,\phi,t;\mathcal{D}_\xi)e^{i\left(1-\frac{2j}{n}\right)w}\Big|_{w=0}
\\
&
=
i\frac{M}{2\pi \hbar t}e^{\frac{iM}{2 \hbar t}r^2}\sum_{m=0}^{\infty} \frac{\left(i e^{i\pi/4}\right)^m}{m!}
\sum_{\ell=0}^m\binom{m}{\ell} c_{m,\ell}(r,\phi,t)
\mathcal{G}_{m,\ell}(\mathcal{D}_w) \mathcal{H}_{\ell}(\mathcal{D}_w)
e^{i\left(1-\frac{2j}{n}\right)w}\Big|_{w=0}.
\end{split}
\end{equation*}
By linearity, the evolution of the superoscillatory initial datum is given by
$$
\Psi_n(r,\phi,t) =\sum_{j=0}^nC_j(n,a)\psi_{g\left(1-\frac{2j}{n}\right), h\left(1-\frac{2j}{n}\right)}(r,\phi,t).
$$
By Theorem \ref{CONTI}, it is immediate that
\[
\begin{split}
&\lim_{n\to\infty}\Psi_n(r,\phi,t)=\lim_{n\to\infty} \sum_{j=0}^nC_j(n,a)\psi_{g\left(1-\frac{2j}{n}\right), h\left(1-\frac{2j}{n}\right)}(r,\phi,t)
\\
&=
i\frac{M}{2\pi \hbar t}e^{\frac{iM}{2 \hbar t}r^2}\sum_{m=0}^{\infty} \frac{\left(i e^{i\pi/4}\right)^m}{m!}
\sum_{\ell=0}^m\binom{m}{\ell} c_{m,\ell}(r,\phi,t)
\mathcal{G}_{m,\ell}(\mathcal{D}_w) \mathcal{H}_{\ell}(\mathcal{D}_w)
 \lim_{n\to\infty} \sum_{j=0}^nC_j(n,a)e^{i\left(1-\frac{2j}{n}\right)w}
\Big|_{w=0}
\\
&
=
i\frac{M}{2\pi \hbar t}e^{\frac{iM}{2 \hbar t}r^2}\sum_{m=0}^{\infty} \frac{\left(i e^{i\pi/4}\right)^m}{m!}
\sum_{\ell=0}^m\binom{m}{\ell} c_{m,\ell}(r,\phi,t)
\mathcal{G}_{m,\ell}(\mathcal{D}_w) \mathcal{H}_{\ell}(\mathcal{D}_w)
 e^{iaw}
\Big|_{w=0}
\\
&
=\psi_{g(a), h(a)}(r,\phi,t)
\end{split}
\]
since the sequence $\displaystyle\sum_{j=0}^nC_j(n,a)e^{i\left(1-\frac{2j}{n}\right)w}$ converges to $e^{ia w}$ in $A_1$, concluding the proof of the Theorem.
 \end{proof}


\begin{thebibliography}{99}


\bibitem{aav} Y. Aharonov, D. Albert, L. Vaidman, {\em How the result of a measurement of a component of the spin of a spin-1/2 particle can turn out to be 100}, Phys. Rev. Lett., {\bf 60} (1988), 1351-1354.


\bibitem{ABCS1}
Y. Aharonov, J. Behrndt, F. Colombo, P. Schlosser,
{\em Schr\"odinger evolution of superoscillations with $\delta$- and $\delta'$-potentials},
 Quantum Stud. Math. Found., {\bf 7} (2020), 293--305.

\bibitem{JDE}
Y. Aharonov, J. Behrndt, F. Colombo, P. Schlosser,
{\em Green's Function for the Schr\"odinger Equation with
 a Generalized Point Interaction and Stability of Superoscillations},
  J. Differential Equations, {\bf 277} (2021), 153--190.


\bibitem{ABCS}
 Y. Aharonov, J. Behrndt, F. Colombo, P. Schlosser,
 {\em A unified approach to Schr\"odinger evolution of superoscillations and supershifts},
 J. Evol. Equ. 22 (2022), no. 1, Paper No. 26, 31 pp.

\bibitem{NEWMETH} Y. Aharonov, F. Colombo, I. Sabadini, T. Shushi, D. C. Struppa, J. Tollaksen,
{\em A new method to generate superoscillating functions and supershifts}, Proc. R. Soc. A 477 (2021).

\bibitem{STEP} Y. Aharonov, F. Colombo, I. Sabadini, D. C. Struppa, J. Tollaksen,
{\em How superoscillating tunneling waves can overcome the step potential},
 Ann. Physics 414 (2020) 168088.



\bibitem{tre}
 Y. Aharonov, F. Colombo,  I. Sabadini, D.C. Struppa, J. Tollaksen,
  {\em  Evolution of superoscillations in the Klein--Gordon field},
   Milan J. Math., {\bf 88} (2020), no. 1, 171--189.



\bibitem{acsst3}
Y. Aharonov,  F. Colombo,  I. Sabadini, D.C. Struppa, J. Tollaksen,
{\em On the Cauchy problem for the Schr\"{o}dinger equation with superoscillatory initial data},
J. Math. Pures Appl., {\bf 99} (2013), 165--173.






\bibitem{acsst5}
Y. Aharonov,  F. Colombo,  I. Sabadini, D.C. Struppa, J. Tollaksen,
{\em The mathematics of superoscillations},  Mem. Amer. Math. Soc., {\bf 247} (2017), no. 1174, v+107 pp.


\bibitem{QS1}
 Y. Aharonov, F. Colombo,  D.C. Struppa,  J.  Tollaksen,
{\em
Schr\"odinger evolution of superoscillations under different potentials},
  Quantum Stud. Math. Found., {\bf 5} (2018), 485--504.



\bibitem{abook} Y. Aharonov, D. Rohrlich, {\em Quantum Paradoxes: Quantum Theory for the Perplexed}, Wiley-VCH Verlag, Weinheim, 2005.


\bibitem{QS3}
    Y. Aharonov, I. Sabadini, J. Tollaksen, A. Yger,
 {\em Classes of superoscillating functions},
     Quantum Stud. Math. Found., {\bf 5} (2018), 439--454.

\bibitem{AShushi} Y. Aharonov, T. Shushi, {\em A new class of superoscillatory functions based on a generalized polar coordinate system}, Quantum Stud. Math. Found., {\bf 7} (2020), 307--313.



\bibitem{QS2}
T. Aoki, F. Colombo, I. Sabadini, D. C. Struppa, {\em
Continuity of some operators arising in the theory of superoscillations},
 Quantum Stud. Math. Found., {\bf 5} (2018),  463--476.



\bibitem{AOKI} T. Aoki, F. Colombo,  I. Sabadini, D.C. Struppa,
{\em Continuity theorems for a class of convolution
operators and applications to superoscillations},
Ann. Mat. Pura Appl.,  {\bf 197} (2018), 1533--1545.

\bibitem{uno}
 D. Alpay, F. Colombo, I. Sabadini, D.C. Struppa, {\em Aharonov-Berry superoscillations in the radial harmonic oscillator potential}, Quantum Stud. Math. Found., {\bf 7} (2020),  269--283.

\bibitem{DIKI}
D. Alpay, F. Colombo, K. Diki, I.  Sabadini,  D. C. Struppa,
 {\em Superoscillations and Fock spaces}. J. Math. Phys. 64 (2023), no. 9, Paper No. 093505, 20 pp.

\bibitem{Jussi}
J. Behrndt, F. Colombo, P. Schlosser, {\em  Evolution of Aharonov--Berry superoscillations in Dirac $\delta$-potential},
 Quantum Stud. Math. Found., {\bf 6} (2019), 279--293.

\bibitem{BOREL} J. Behrndt, F. Colombo, P. Schlosser, D.C. Struppa, {\em Integral representation of superoscillations via complex Borel measures and their convergence}, Trans. Amer. Math. Soc. 376 (2023), no. 9, 6315-6340.

\bibitem{BG_book}
 C. A. Berenstein, R. Gay, {\em Complex Analysis and Special Topics in Harmonic Analysis}, Springer-Verlag, New York, 1995.

\bibitem{Be19} M. Berry et al, \textit{Roadmap on superoscillations}, 2019, Journal of Optics 21 053002.

\bibitem{berry2} M. V. Berry, {\em Faster than Fourier}, in
Quantum Coherence and Reality; in celebration of the 60th Birthday of
Yakir Aharonov ed. J. S. Anandan and J. L. Safko, World Scientific,
Singapore, (1994), pp. 55-65.



\bibitem{berry-noise-2013}  M. Berry, {\em Exact nonparaxial transmission of subwavelength detail using superoscillations}, {J. Phys. A} {\bf 46}, (2013), 205203.

\bibitem{BerryMILAN}
 M. V. Berry,
 {\em Representing superoscillations and narrow Gaussians with elementary functions},
  Milan J. Math., {\bf 84} (2016),  217--230.

\bibitem{SPIN} F. Colombo, E. Pozzi, I. Sabadini, B. D. Wick,
{\em Evolution of superoscillations for spinning particles},
Proc. Amer. Math. Soc. Ser. B 10 (2023), 129-143.

\bibitem{Talbot}
F. Colombo, I. Sabadini,  D.C. Struppa, A. Yger,
{\em
Gauss sums, superoscillations and the Talbot carpet},
 J. Math. Pures Appl., (9) {\bf 147} (2021), 163--178.

 \bibitem{REGULAR-SAMP}
F. Colombo, I. Sabadini, D. C. Struppa, A. Yger,
{\em Analyticity and supershift with regular sampling},
Preprint arXiv:2310.11528.


\bibitem{IRREGULAR-SAMP}
F. Colombo, I. Sabadini, D. C. Struppa, A. Yger,
{\em Analyticity and supershift with irregular sampling},
arXiv:2312.05089, to apper Complex Analysis and its Synergies.




\bibitem{hyper}
F. Colombo, I. Sabadini,  D.C. Struppa, A. Yger,
{\em Superoscillating sequences and hyperfunctions},
 Publ. Res. Inst. Math. Sci., {\bf 55} (2019), no. 4, 665--688.



\bibitem{due}
F. Colombo, G. Valente, { \em Evolution of Superoscillations in the Dirac Field},
 Found. Phys., {\bf 50} (2020), 1356--1375.


\bibitem{kempf1}
P. J. S. G. Ferreira, A. Kempf,
{\em Unusual properties of superoscillating particles},
{J. Phys. A}, {\bf 37} (2004), 12067-76.


\bibitem{kempf2}
P. J. S. G. Ferreira, A. Kempf,
{\em Superoscillations: faster than the Nyquist rate},
{IEEE Trans. Signal Processing}, {\bf 54} (2006), 3732--3740.

\bibitem{kempf2HHH}
 P. J. S. G. Ferreira, A. Kempf,  M. J. C. S. Reis,
 {\em Construction of Aharonov-Berry's superoscillations}, J. Phys. A {\bf 40} (2007),  5141-5147.

\bibitem{Feynman}
 Feynman, R. P. Space-time approach to non-relativistic quantum mechanics. Rev. Modern Physics 20 (1948), 367-387.


\bibitem{FeynmanHIBBS}
R. P. Feynman, A. R. Hibbs,
{\em Quantum mechanics and path integrals}.
Emended edition. Emended and with a preface by Daniel F. Styer. Dover Publications, Inc., Mineola, NY, 2010. xii+371 pp.



\bibitem{QS20} A.N. Jordan, {\em Superresolution using supergrowth and intensity contrast imaging},
 Quantum Stud.: Math. Found., 7(3) (2020) 285-292.


\bibitem{QUANOB} A. N. Jordan, Y. Aharonov, D. C. Struppa, F. Colombo, I. Sabadini, T. Shushi, J. Tollaksen, J. C. Howell, A. N. Vamivakas,
    {\em Super-phenomena in arbitrary quantum observables}, arXiv:2209.05650.

\bibitem{PHREVL} J.C. Howell, A.N. Jordan, B. Šoda, A. Kempf,
{\em Super Interferometric Range Resolution}, Phys. Rev. Lett., 131(5) (2023) p.053803.


\bibitem{OPTICS} F.M. Huang, Y. Chen, F.J.G. De Abajo, N.I. Zheludev,
{\em Optical super-resolution through super-oscillations},
J. Optics A: Pure Appl.Optics, 9(9) (2007) p.285.




\bibitem{kempfQS}
A. Kempf,
{\em Four aspects of superoscillations},
 Quantum Stud. Math. Found., {\bf 5} (2018),  477-484.


\bibitem{LaidlawMoretteDeWitt}
M. G. G. Laidlaw, C. Morette-DeWitt,
{\em Feynman Functional Integrals for Systems of Indistinguishable Particles},
Phys. Rev. D 3, 1375 - Published 15 March 1971


\bibitem{lindberg} J. Lindberg, {\em Mathematical concepts of optical superresolution},  Journal of Optics, {\bf 14} (2012),
 083001.

 \bibitem{MAYS}
 J. Mays, G. Gbur, {\em Partially coherent superoscillations in the Talbot effect},
  J. Phys. A 55 (2022), no. 50, Paper No. 504002, 14 pp.


\bibitem{Morandi}
G. Morandi, E. Menossi,
{\em  Path-integrals in multiply-connected spaces and the Aharonov-Bohm effect}.
European J. Phys. 5 (1984), no. 1, 49-58.


\bibitem{Pozzi}
E. Pozzi,  B. D. Wick, {\em Persistence of superoscillations under the Schr\"odinger equation},
 Evol. Equ. Control Theory, 11 (2022), no. 3, 869-894.


\bibitem{peter}
 P. Schlosser,  {\em Time evolution of superoscillations for the Schr\"odinger equation on $\mathbb{R}\setminus \{0\}$},
  Quantum Stud. Math. Found. 9 (2022), no. 3, 343-366.

\bibitem{sodakemp}
B. Soda,  A. Kempf, {\em  Efficient method to create superoscillations with generic target behavior},
Quantum Stud. Math. Found., {\bf 7} (2020), no. 3, 347--353.


\bibitem{SchulmanBOOK}
L. S. Schulman, {\em Techniques and applications of path integration}. A Wiley-Interscience Publication. John Wiley and Sons, Inc., New York, 1981. xv+359 pp.




\bibitem{Watson}
 Watson, G. N., {\em A Treatise on the theory of Bessel functions}. Reprint of the second (1944) edition. Cambridge Mathematical Library. Cambridge University Press, Cambridge, 1995. viii+804 pp


\end{thebibliography}
\end{document}